\documentclass[]{eptcs}
\providecommand{\event}{QAPL 2012}
\usepackage{proof, amsmath, amssymb, latexsym, enumerate}

\usepackage[all]{xy}
\usepackage{epsfig}
\usepackage{float}
\usepackage{graphicx}

\floatstyle{boxed}
\newfloat{boxedfig}{thp}{lop}
\floatname{boxedfig}{Figure}

\newcommand{\texcomment}[1]{}

\newcommand{\aset}[1]{\{{#1}\}}

\newtheorem{theorem}{Theorem}[section]
\newtheorem{definition}[theorem]{Definition}
\newtheorem{lemma}[theorem]{Lemma}

\newenvironment{proof}{\noindent\rm{\bf Proof:}}{\hbox{$\Box$}\vspace*{0.2\baselineskip}}

\newcommand{\addappendix}[2]{#2}

\newcommand{\ttif}[3]{\texttt{if}\;{#1}\;\texttt{then}\;{#2}\;\texttt{else}\;{#3}}
\newcommand{\ttassign}[2]{{#1}\;\texttt{:=}\;{#2}}

\newcommand{\ttskip}{\texttt{skip}}

\newcommand{\vect}[1]{\overrightarrow{{#1}}}

\begin{document}

\title{Quantitative Information Flow as Safety and Liveness
  Hyperproperties\thanks{This work was supported by MEXT KAKENHI
    23700026, 22300005, 23220001, and Global COE Program ``CERIES.''}}

\author{Hirotoshi Yasuoka
\institute{Tohoku University\\
Sendai, Japan}
\email{yasuoka@kb.ecei.tohoku.ac.jp}
\and
Tachio Terauchi
\institute{Nagoya University\\
Nagoya, Japan}
\email{\quad terauchi@is.nagoya-u.ac.jp}
}

\def\titlerunning{Quantitative Information Flow as Safety and Liveness
  Hyperproperties}
\def\authorrunning{H. Yasuoka \& T. Terauchi}
\maketitle

\begin{abstract}
  We employ Clarkson and Schneider's ``hyperproperties'' to classify
  various verification problems of quantitative information flow.  The
  results of this paper unify and extend the previous results on the
  hardness of checking and inferring quantitative information flow.
  In particular, we identify a subclass of liveness hyperproperties,
  which we call ``$k$-observable hyperproperties'', that can be
  checked relative to a reachability oracle via self composition.
\end{abstract}

\section{Introduction}
\label{sec:introduction}

We consider programs containing high security inputs and low security
outputs.  Informally, the quantitative information flow problem
concerns the amount of information that an attacker can learn about
the high security input by executing the program and observing the low
security output.  The problem is motivated by applications in
information security.  We refer to the classic by
Denning~\cite{denning82} for an overview.

In essence, quantitative information flow measures {\em how} secure,
or insecure, a program (or a part of a program --e.g., a variable--)
is.  Thus, unlike
non-interference~\cite{DBLP:conf/sosp/Cohen77,goguen:sp1982}, that
only tells whether a program is completely secure or not completely
secure, a definition of quantitative information flow must be able to
distinguish two programs that are both interfering but have different
levels of security.

For example, consider the programs $M_1 \equiv \ttif{H =
  g}{\ttassign{O}{0}}{\ttassign{O}{1}}$ and $M_2 \equiv
\ttassign{O}{H}$.  In both programs, $H$ is a high security input and
$O$ is a low security output.  Viewing $H$ as a password, $M_1$ is a
prototypical login program that checks if the guess $g$ matches the
password.  By executing $M_1$, an attacker only learns whether $H$ is
equal to $g$, whereas she would be able to learn the entire content of
$H$ by executing $M_2$.  Hence, a reasonable definition of
quantitative information flow should assign a higher quantity to $M_2$
than to $M_1$, whereas non-interference would merely say that $M_1$
and $M_2$ are both interfering, assuming that there are more than one
possible value of $H$.

Researchers have attempted to formalize the definition of quantitative
information flow by appealing to information theory.  This has
resulted in definitions based on the Shannon
entropy~\cite{denning82,clarkjcs2007,malacaria:popl2007}, the min
entropy~\cite{smith09}, and the guessing
entropy~\cite{kopf07,DBLP:conf/sp/BackesKR09}.  All of these
definitions map a program (or a part of a program) onto a non-negative
real number, that is, they define a function $\mathcal{X}$ such that
given a program $M$, $\mathcal{X}(M)$ is a non-negative real number.
(Concretely, $\mathcal{X}$ is ${\it SE}[\mu]$ for the
Shannon-entropy-based definition with the distribution $\mu$, ${\it
  ME}[\mu]$ for the min-entropy-based definition with the distribution
$\mu$, and ${\it GE}[\mu]$ for the guessing-entropy-based definition
with the distribution $\mu$.)

In a previous
work~\cite{DBLP:conf/csfw/YasuokaT10,DBLP:conf/esorics/YasuokaT10}, we
have proved a number of hardness results on checking and inferring
quantitative information flow (QIF) according to these definitions.  A
key concept used to connect the hardness results to QIF verification
problems was the notion of $k$-safety, which is an instance in a
collection of the class of program properties called {\em
  hyperproperties}~\cite{DBLP:journals/jcs/ClarksonS10}.  In this paper,
we make the connection explicit by providing a fine-grained
classification of QIF problems, utilizing the full range of
hyperproperties.  This has a number of benefits, summarized below.
\begin{itemize}
\item[1.)] A unified view on the hardness results of QIF problems.
\item[2.)] New insights into hyperproperties themselves.
\item[3.)] A straightforward derivation of some complexity theoretic results.
\end{itemize}
Regarding 1.), we focus on two types of QIF problems, an
upper-bounding problem that checks if QIF of a program is bounded
above by the given number, and a lower-bounding problem that checks if
QIF is bounded below by the given number.  Then, for each QIF
definitions {\it SE}, {\it GE}, {\it ME}, we classify whether or not
they are safety hyperproperty, $k$-safety hyperproperty, liveness
hyperproperty, or $k$-observable hyperproperty (and give a bound on
$k$ for $k$-safe/$k$-observable).  Safety hyperproperty, $k$-safety
hyperproperty, liveness hyperproperty, and observable hyperproperty
are classes of hyperproperties defined by Clarkson and
Schneider~\cite{DBLP:journals/jcs/ClarksonS10}.  In this paper, we
identify new classes of hyperproperties, $k$-observable hyperproperty,
that is useful for classifying QIF problems.  $k$-observable
hyperproperty is a subclass of observable hyperproperties, and
observable hyperproperty is a subclass of liveness
hyperproperties.\footnote{Technically, only non-empty observable
  hyperproperties are liveness hyperproperties.}  We focus on the case
the input distribution is uniform, that is, $\mu = U$, as showing the
hardness for a specific case amounts to showing the hardness for the
general case.  Also, checking and inferring QIF under the
uniformly distributed inputs has received much
attention~\cite{DBLP:conf/ifip1-7/HeusserM09,DBLP:conf/sp/BackesKR09,DBLP:conf/csfw/KopfR10,clark05,malacaria:popl2007,clarkjcs2007},
and so, the hardness for the uniform case is itself of research
interest.\footnote{In fact, computing QIF under other input
  distributions can sometimes be reduced to this
  case~\cite{DBLP:conf/ccs/SamaratiV10}.  See also Section~\ref{sec:maxqif}.}

Regarding 2.), we show that the $k$-observable subset of the
observable hyperproperties is amenable to verification via self
composition~\cite{barthe:csfw04,darvas:spc05,terauchi:sas05,naumann:esorics06,unno:plas2006},
much like $k$-safety hyperproperties, and identify which QIF problems
belong to that family.  We also show that the hardest of the QIF
problems (but nevertheless one of the most popular) can only be
classified as a general liveness hyperproperty, suggesting that
liveness hyperproperty is a quite permissive class of hyperproperties.

Regarding 3.), we show that many complexity theoretic results for QIF
problems of loop-free boolean programs can be derived from their
hyperproperties
classifications~\cite{DBLP:conf/csfw/YasuokaT10,DBLP:conf/esorics/YasuokaT10}.
We also prove new complexity theoretic results, including the
(implicit state) complexity results for loop-ful boolean programs,
complementing the recently proved explicit state complexity
results~\cite{DBLP:conf/csfw/CernyCH11}.

Table~\ref{fig:hysummary} and Table~\ref{fig:comsummary} summarize the
hyperproperties classifications and computational complexities of
upper/lower-bounding problems.  We abbreviate lower-bounding problem,
upper-bounding problem, and boolean programs to LBP, UBP, and BP,
respectively.  The ``constant bound'' rows correspond to bounding
problems with a constant bound (whereas the plain bounding problems
take the bound as an input).
\begin{table}
\begin{center}
  \begin{tabular}{|c|c|c|c|}
    \hline
    & ${\it SE}[U]$ & ${\it ME}[U]$ & ${\it GE}[U]$\\
    \hline
    LBP& Liveness & Liveness & Liveness\\
    \hline
    UBP& Liveness & Safety & Safety\\
    \hline
    LBP constant bound& Liveness & $k$-observable & $k$-observable\\
    \hline
    UBP constant bound& Liveness & $k$-safety~\cite{DBLP:conf/esorics/YasuokaT10} & $k$-safety~\cite{DBLP:conf/esorics/YasuokaT10}\\
    \hline
\end{tabular}
\end{center}
  \caption{A summary of hyperproperty classifications}
\label{fig:hysummary}
\end{table}
\begin{table}
\begin{center}
  \begin{tabular}{|c|c|c|c|}
    \hline
    & ${\it SE}[U]$ & ${\it ME}[U]$ & ${\it GE}[U]$\\
    \hline
    LBP for BP& {\it PSPACE}-hard& {\it PSPACE}-complete& {\it PSPACE}-complete\\
    \hline
    UBP for BP& {\it PSPACE}-hard & {\it PSPACE}-complete& {\it PSPACE}-complete\\
    \hline
    LBP for loop-free BP& {\it PP}-hard & {\it PP}-hard & {\it PP}-hard\\
    \hline
    UBP for loop-free BP& {\it PP}-hard~\cite{DBLP:conf/esorics/YasuokaT10} & {\it PP}-hard~\cite{DBLP:conf/esorics/YasuokaT10} & {\it PP}-hard~\cite{DBLP:conf/esorics/YasuokaT10}\\
    \hline
    LBP for loop-free BP, constant bound& Unknown & {\it NP}-complete & {\it NP}-complete\\
    \hline
    UBP for loop-free BP, constant bound& Unknown & {\it coNP}-complete & {\it coNP}-complete\\
    \hline
\end{tabular}
\end{center}
\caption{A summary of computational complexities}
\label{fig:comsummary}
\end{table}

\addappendix{The proofs omitted from the main body of the paper appear
  in the Appendix.}{The proofs omitted from the paper appear in the
  extended report~\cite{longversion}.}

\section{Preliminaries}

\subsection{Quantitative Information Flow}

We introduce the information theoretic definitions of QIF that have
been proposed in literature.  First, we review the notion of the {\em
  Shannon entropy}~\cite{shannon48}, $\mathcal{H}[\mu](X)$, which is
the average of the information content, and intuitively, denotes the
uncertainty of the random variable $X$.  And, we review the notion of
the {\em conditional entropy}, $\mathcal{H}[\mu](Y|Z)$, which denotes
the uncertainty of $Y$ after knowing $Z$.
\begin{definition}[Shannon Entropy and Conditional Entropy]
  Let $X$ be a random variable with sample space $\mathbb X$ and $\mu$
  be a probability distribution associated with $X$ (we write $\mu$
  explicitly for clarity).  The Shannon entropy of $X$ is defined as
\[
\mathcal{H}[\mu](X)=\sum_{x\in\mathbb{X}} \mu(X=x)\log\frac{1}{\mu(X=x)}
\]

Let $Y$ and $Z$ be random variables with sample space $Y$ and $Z$,
respectively, and $\mu'$ be a probability distribution associated with
$Y$ and $Z$.  Then, the conditional entropy of $Y$ given $Z$ is
defined as
\[
\mathcal{H}[\mu](Y|Z)=\sum_{z\in\mathbb Z} \mu(Z=z) \mathcal{H}[\mu](Y|Z=z)
\]
where
\[
\begin{array}{l}
\mathcal{H}[\mu](Y|Z=z)=\sum_{y\in\mathbb Y} \mu(Y=y|Z=Z)\log\frac{1}{\mu(Y=y|Z=z)} \\
  \mu(Y=y|Z=z)=\frac{\mu(Y=y,Z=z)}{\mu(Z=z)}
\end{array}
\]
(The logarithm is in base 2.)
\end{definition}

Let $M$ be a program that takes a high security input $H$, and gives
the low security output trace $O$.  For simplicity, we restrict to
programs with just one variable of each kind, but it is trivial to
extend the formalism to multiple variables (e.g., by letting the
variables range over tuples or lists).  Also, for the purpose of the
paper, unobservable (i.e., high security) output traces are
irrelevant, and so we assume that the only program output is the low
security output trace.  Let $\mu$ be a probability distribution over
the values of $H$.  Then, the semantics of $M$ can be defined by the
following probability equation.  (We restrict to deterministic
programs in this paper.)
\[
\mu(O = o) = \sum_{\scriptsize \begin{array}{l}h \in \mathbb{H}\\ M(h) = o\end{array}} \mu(H = h)
\]
Here, $M(h)$ denotes the infinite low security output trace of the
program $M$ given a input $h$, and $M(h)=o$ denotes the output trace
of $M$ given $h$ that is equivalent to $o$.  In this paper, we adopt
the termination-insensitive security observation model, and let the
outputs $o$ and $o'$ be equivalent iff $\forall i \in \omega.o_i=\bot
\vee o_i'=\bot \vee o_i=o_i'$ where $o$ and $o_i$ denotes the $i$th
element of $o$, and $\bot$ is the special symbol denoting
termination.\footnote{Here, we adopt the trace based QIF formalization
  of~\cite{malacaria08}.}

In this paper, programs are represented by sets of traces, and traces
are represented by lists of stores of programs.  More formally,
\[
\begin{array}{l}
  M(h)\;\text{is equal to}\;o \quad\text{iff}\quad  \sigma_0;\sigma_1;\ldots;\sigma_i;\ldots \in M\\
  \qquad\text{where}\;\sigma_0(H)=h\;\text{and}\;\forall i\geq 1.\sigma_i(O)=o_i\;(o_i\;\text{denotes the ith element of}\;o)
\end{array}
\]
Here, $\sigma$ denotes a store that maps variables to values.  Because
we restrict all programs to deterministic programs, every program $M$
satisfies the following condition: For any trace
$\vect\sigma,\vect\sigma'\in M$, we have
$\sigma_0(H)=\sigma_0'(H)\Rightarrow \vect\sigma=\vect\sigma'$ where
$\sigma_0$ and $\sigma_0'$ denote the first elements of $\vect\sigma$
and $\vect\sigma'$, respectively.  Now, we are ready to define
Shannon-entropy-based quantitative information
flow.
\begin{definition}[Shannon-Entropy-based
  QIF~\cite{denning82,clarkjcs2007,malacaria:popl2007}]
\label{def:se}
Let $M$ be a program with a high security input $H$, and a low
security output trace $O$.  Let $\mu$ be a distribution over $H$.
Then, the Shannon-entropy-based quantitative information flow is
defined
\[
{\it SE}[\mu](M) =  \mathcal{H}[\mu](H)-\mathcal{H}[\mu](H|O)
\]
\end{definition}
Intuitively, $\mathcal{H}[\mu](H)$ denotes the initial uncertainty and
$\mathcal{H}[\mu](H|O)$ denotes the remaining uncertainty after
knowing the low security output trace.  (For space, the paper focuses
on the low-security-input free definitions of QIF.)  

As an example, consider the programs $M_1$ and $M_2$ from
Section~\ref{sec:introduction}.  For concreteness, assume that $g$ is
the value $01$ and $H$ ranges over the space $\aset{00,01,10,11}$.
Let $U$ be the uniform distribution over $\aset{00,01,10,11}$, that
is, $U(h) = 1/4$ for all $h \in \aset{00,01,10,11}$.  The results are
as follows.
\[
\begin{array}{rcl}
{\it SE}[U](M_1)&=&\mathcal{H}[U](H)-\mathcal{H}[U](H|O)=\log 4-\frac{3}{4}\log{3}\approx .81128\\&&\\

{\it SE}[U](M_2)&=&\mathcal{H}[U](H)-\mathcal{H}[U](H|O)=\log 4-\log 1=2
\end{array}
\]
Consequently, we have that ${\it SE}[U](M_1) \leq {\it SE}[U](M_2)$,
but ${\it SE}[U](M_2) \not\leq {\it SE}[U](M_1)$.  That is, $M_1$ is
more secure than $M_2$ (according to the Shannon-entropy based
definition with uniformly distributed inputs), which agrees with our
intuition.

Next, we introduce the {\em min entropy}, which has recently been
suggested as an alternative measure for quantitative information
flow~\cite{smith09}.
\begin{definition}[Min Entropy]
Let $X$ and $Y$ be random variables, and $\mu$ be an associated probability
distribution.  Then, the min entropy of $X$ is defined
\[
\mathcal{H}_\infty[\mu](X)=\log\frac{1}{\mathcal{V}[\mu](X)}
\]
and the conditional min entropy of $X$ given $Y$ is defined
\[
\mathcal{H}_\infty[\mu](X|Y)=\log\frac{1}{\mathcal{V}[\mu](X|Y)}
\]
where
\[
\begin{array}{rcl}
\mathcal{V}[\mu](X)&=&\max_{x\in\mathbb X} \mu(X=x)\\
\mathcal{V}[\mu](X|Y=y)&=&\max_{x\in\mathbb X} \mu(X=x|Y=y)\\
\mathcal{V}[\mu](X|Y)&=&\sum_{y\in\mathbb Y} \mu(Y=y) \mathcal{V}[\mu](X|Y=y)
\end{array}
\]
\end{definition}

Intuitively, $\mathcal{V}[\mu](X)$ represents the highest probability
that an attacker guesses $X$ in a single try.  We now define the
min-entropy-based definition of QIF.

\begin{definition}[Min-Entropy-based QIF~\cite{smith09}]
\label{def:me}
Let $M$ be a program with a high security input $H$, and a low
security output trace $O$.  Let $\mu$ be a distribution over $H$.
Then, the min-entropy-based quantitative information flow is defined
\[ 
{\it ME}[\mu](M)=\mathcal{H}_\infty[\mu](H)-\mathcal{H}_\infty[\mu](H|O)
\]
\end{definition}

Computing the min-entropy based quantitative information flow for our
running example programs $M_1$ and $M_2$ from
Section~\ref{sec:introduction} with the uniform distribution, we
obtain,
\[
\begin{array}{rcl}
{\it ME}[U](M_1)&=&\mathcal{H}_\infty[U](H)-\mathcal{H}_\infty[U](H|O)=\log 4-\log 2=1\\&&\\
{\it ME}[U](M_2)&=&\mathcal{H}_\infty[U](H)-\mathcal{H}_\infty[U](H|O)=\log 4 -\log 1=2
\end{array}
\]
Again, we have that ${\it ME}[U](M_1) \leq {\it ME}[U](M_2)$ and ${\it
  ME}[U](M_2) \not\leq {\it ME}[U](M_1)$, and so $M_2$ is deemed less
secure than $M_1$.

The third definition of quantitative information flow treated in this
paper is the one based on the guessing entropy~\cite{Massey94}, that
has also recently been proposed in
literature~\cite{kopf07,DBLP:conf/sp/BackesKR09}.
\begin{definition}[Guessing Entropy]
Let $X$ and $Y$ be random variables, and $\mu$ be an associated probability
distribution.  Then, the guessing entropy of $X$ is defined
\[
\mathcal{G}[\mu](X)=\sum_{1\le i\le m}i\times\mu(X=x_i)
\]
where $\aset{x_1,x_2,\dots,x_{m}} = \mathbb{X}$ and
$\forall i,j.i\le j\Rightarrow \mu(X=x_i)\ge\mu(X=x_j)$.

The conditional guessing entropy of $X$ given $Y$ is defined
\[
\mathcal{G}[\mu](X|Y)=\sum_{y\in{\mathbb Y}}\mu(Y=y)\sum_{1\le i\le m}i\times\mu(X=x_i|Y=y)
\]
where $\aset{x_1,x_2,\dots,x_{m}} = \mathbb{X}$ and $\forall i,j.i\le
j\Rightarrow \mu(X=x_i|Y=y)\ge\mu(X=x_j|Y=y)$.
\end{definition}

Intuitively, $\mathcal{G}[\mu](X)$ represents the average number of
times required for the attacker to guess the value of $X$.  We now
define the guessing-entropy-based quantitative information flow.

\begin{definition}[Guessing-Entropy-based
  QIF~\cite{kopf07,DBLP:conf/sp/BackesKR09}]
\label{def:ge}
Let $M$ be a program with a high security input $H$, and a low
security output trace $O$.  Let $\mu$ be a distribution over $H$.
Then, the guessing-entropy-based quantitative information flow is
defined
\[
{\it GE}[\mu](M)=\mathcal{G}[\mu](H)-\mathcal{G}[\mu](H|O)
\]
\end{definition}

We test {\it GE} on the running example from
Section~\ref{sec:introduction} by calculating the quantities for the
programs $M_1$ and $M_2$ with the uniform distribution.  
\[
\begin{array}{rcl}
{\it GE}[U](M_1) &=&\mathcal{G}[U](H)-\mathcal{G}[U](H|O)=\frac{5}{2} - \frac{7}{4} =  0.75\\\\
{\it GE}[U](M_2)& = &\mathcal{G}[U](H)-\mathcal{G}[U](H|O)=\frac{5}{2} - 1 = 1.5
\end{array}
\]
Therefore, we again have that ${\it GE}[U](M_1) \leq {\it GE}[U](M_2)$
and ${\it GE}[U](M_2) \not\leq {\it GE}[U](M_1)$, and so $M_2$ is
considered less secure than $M_1$, even with the guessing-entropy
based definition with the uniform distribution.

\subsection{Bounding Problems}
\label{subsec:bp}
We introduce the bounding problems of quantitative information flow
that we classify.  First, we define the QIF upper-bounding problems.
Upper-bounding problems are defined as follows: Given a program $M$
and a rational number $q$, decide if the information flow of $M$ is
less than or equal to $q$.
\[
\begin{array}{c}
  {\mathcal U}_{\it SE}=\aset{(M,q)\mid {\it SE}[U](M)\le q}\\
  {\mathcal U}_{\it ME}=\aset{(M,q)\mid {\it ME}[U](M)\le q}\\
  {\mathcal U}_{\it GE}=\aset{(M,q)\mid {\it GE}[U](M)\le q}
\end{array}
\]
Recall that $U$ denotes the uniform distribution.

Next, we define lower-bounding problems.  Lower-bounding problems are
defined as follows: Given a program $M$ and a rational number $q$,
decide if the information flow of $M$ is greater than $q$.
\[
\begin{array}{c}
  {\mathcal L}_{\it SE}=\aset{(M,q)\mid {\it SE}[U](M)> q}\\
  {\mathcal L}_{\it ME}=\aset{(M,q)\mid {\it ME}[U](M)> q}\\
  {\mathcal L}_{\it GE}=\aset{(M,q)\mid {\it GE}[U](M)> q}
\end{array}
\]

\subsection{Non Interference}

We recall the notion of non-interference, which, intuitively, says that
the program leaks no information.
\begin{definition}[Non-intereference~\cite{DBLP:conf/sosp/Cohen77,goguen:sp1982}]
  A program $M$ is said to be non-interfering iff for any $h,h'\in
  \mathbb{H}$, $M(h)=M(h')$.
\end{definition}

Non-interference is known to be a special case of bounding problems that
tests against $0$.
\begin{theorem}[\cite{clark05,DBLP:conf/esorics/YasuokaT10}]
\label{thm:nonint}
1.) $M$ is non-interfering iff $(M,0)\in\mathcal{U}_{\it SE}$.
2.) $M$ is non-interfering iff $(M,0)\in\mathcal{U}_{\it ME}$.
3.) $M$ is non-interfering iff $(M,0)\in\mathcal{U}_{\it GE}$.
\end{theorem}

\section{Liveness Hyperproperties}

Clarkson and Schneider have proposed the notion of
hyperproperties~\cite{DBLP:journals/jcs/ClarksonS10}.  
\begin{definition}[Hyperproperties~\cite{DBLP:journals/jcs/ClarksonS10}]
  We say that $P$ is a hyperproperty if $P\subseteq
  \mathcal{P}(\Psi_{\tt inf})$ where $\Psi_{\tt{inf}}$ is the set of
  all infinite traces, and $\mathcal{P}(X)$ denote the powerset of
  $X$.
\end{definition}
Note that hyperproperties are sets of trace sets.  As such, they are
more suitable for classifying information properties than the
classical trace properties which are sets of traces.  For example,
non-interference is not a trace property but a hyperproperty.

Clarkson and Schneider have identified a subclass of hyperproperties,
called liveness hyperproperties, as a generalization of liveness
properties.  Intuitively, a liveness hyperproperty is a property that
can not be refuted by a finite set of finite traces.  That is, if $P$
is a liveness hyperproperty, then for any finite set of finite traces
$T$, there exists a set of traces that contains $T$ and satisfies $P$.
Formally, let ${\it Obs}$ be the set of finite sets of finite traces,
and ${\it Prop}$ be the set of sets of infinite traces (i.e.,
hyperproperties), that is,
\[
\begin{array}{rcl}
  {\it Obs}&=&\mathcal{P}^{\tt fin}(\Psi_{\tt{fin}})\\
  {\it Prop}&=&\mathcal{P}(\Psi_{\tt{inf}})
\end{array}
\]
(Here, $\mathcal{P}^{\tt{fin}}(X)$ denotes the finite subsets of $X$,
$\Psi_{\tt{fin}}$ denotes the set of finite traces.)  Let $\le$
be the relation over ${\it Obs}\times{\it Prop}$ such that
\[
S \le T\quad\textrm{iff}\quad\forall t\in S.\exists t'. t\circ t'\in T
\]
where $t\circ t'$ is the sequential composition of $t$ and $t'$.
Then,
\begin{definition}[Liveness
  Hyperproperties~\cite{DBLP:journals/jcs/ClarksonS10}]
  We say that a hyperproperty $P$ is a liveness hyperproperty if for
  any set of traces $S\in {\it Obs}$, there exists a set of traces
  $S'\in{\it Prop}$ such that $S\leq S'$ and $S'\in P$.
\end{definition}

Now, we state the first main result of the paper: the lower-bounding
problems are liveness hyperproperties.\footnote{We implicitly extend
  the notion of hyperproperties to classify hyperproperties that take
  programs and rational numbers.
  See~\cite{DBLP:conf/esorics/YasuokaT10}.}  We note that, because QIF
is restricted to that of deterministic programs in this paper, the
results on bounding problems are for hyperproperties of deterministic
systems.\footnote{This is done simply by restricting {\it Obs} and
  {\it Prop} to those of deterministic systems.  See
  \cite{DBLP:journals/jcs/ClarksonS10} for detail.}
\begin{theorem}
\label{thm:Llp}
${\mathcal L}_{\it SE}$, ${\mathcal L}_{\it ME}$, and ${\mathcal
  L}_{\it GE}$ are liveness hyperproperties.
\end{theorem}
The proof follows from the fact that, for any program $M$, there
exists a program $M'$ containing all the observations of $M$ and has
an arbitrary large information flow.\footnote{Here, we assume that the
  input domains are not bounded.  Therefore, we can construct a
  program that leaks more high-security inputs by enlarging the input
  domain.  Hyperproperty classifications of bounding problems with
  bounded domains appear in Section~\ref{sec:bounddomain}.}

We show that the upper-bounding problem for Shannon-entropy based
quantitative information flow is also a liveness hyperproperty.
\begin{theorem}
\label{thm:USElp}
${\mathcal U}_{\it SE}$ is a liveness hyperproperty.
\end{theorem}
The theorem follows from the fact that we can lower the amount of the
information flow by adding traces that have the same output trace.
Therefore, for any program $M$, there exists $M'$ having more
observation than $M$ such that ${\it SE}[U](M') \leq q$.

\subsection{Observable Hyperproperties}
Clarkson and Schneider~\cite{DBLP:journals/jcs/ClarksonS10} have
identified a class of hyperproperties, called {\em observable
  hyperproperties}, to generalize the notion of observable
properties~\cite{DBLP:journals/apal/Abramsky91} to sets of
traces.\footnote{Roughly, an observable property is a set of traces
  having a finite evidence prefix such that any trace having the
  prefix is also in the set.}
\begin{definition}[Observable
  Hyperproperties~\cite{DBLP:journals/jcs/ClarksonS10}]
  We say that $P$ is a observable hyperproperty if for any set of
  traces $S \in P$, there exists a set of traces $T\in{\it Obs}$ such
  that $T\le S$, and for any set of traces $S'\in{\it Prop}$, $T\le
  S'\Rightarrow S'\in P$.
\end{definition}
We call $T$ in the above definition an {\em evidence}.

Intuitively, observable hyperproperty is a property that can be verified
by observing a finite set of finite traces.  We prove a relationship
between observable hyperproperties and liveness hyperproperties.
\begin{theorem}
\label{thm:nonempty}
  Every non-empty observable hyperproperty is a liveness
  hyperproperty.
\end{theorem}
\begin{proof}
  Let $P$ be a non-empty observable hyperproperty.  It must be the
  case that there exists a set of traces $M\in P$.  Then, there exists
  $T\in {\it Obs}$ such that $T\le M$ and $\forall M'\in{\it Prop}.
  T\le M'\Rightarrow M'\in P$.  For any set of traces $S\in{\it Obs}$,
  there exists $M'\in {\it Prop}$ such that $S\le M'$.  Then, we have
  $M\cup M'\in P$, because $T\le M\cup M'$.  Therefore, $P$ is a
  liveness hyperproperty.
\end{proof}\\
We note that the empty set is not a liveness hyperproperty but an
observable hyperproperty.

We show that lower-bounding problems for min-entropy and
guessing-entropy are observable hyperproperties.
\begin{theorem}
\label{thm:LMEsl}
  $\mathcal{L}_{\it ME}$ is an observable hyperproperty.
\end{theorem}
\begin{theorem}
\label{thm:LGEsl}
  $\mathcal{L}_{\it GE}$ is an observable hyperproperty.
\end{theorem}
Theorem~\ref{thm:LMEsl} follows from the fact that, if $(M,q)\in
\mathcal{L}_{\it ME}$, then $M$ contains an evidence of
$\mathcal{L}_{\it ME}$.  This follows from the fact that when a
program $M'$ contains at least as much observation as $M$, ${\it
  ME}[U](M) \leq {\it ME}[U](M')$ (cf. Lemma~\ref{lem:melog}).
Theorem~\ref{thm:LGEsl} is proven in a similar manner.

We show that neither of the bounding problems for Shannon-entropy are
observable hyperproperties.
\begin{theorem}
\label{thm:ULSEnsl}
  Neither $\mathcal{U}_{\it SE}$ nor $\mathcal{L}_{\it SE}$ is an
  observable hyperproperty.
\end{theorem}
We give the intuition of the proof for $\mathcal{U}_{\it SE}$.
Suppose ${\it SE}[U](M)\le q$.  $M$ does not provide an evidence of
${\it SE}[U](M)\le q$, because for any potential evidence, we can
raise the amount of the information flow by adding traces that have
disjoint output traces.  The result for $\mathcal{L}_{\it SE}$ is
shown in a similar manner.

It is interesting to note that the bounding problems of {\it SE} can
only be classified as general liveness hyperproperties (cf.
Theorem~\ref{thm:Llp} and \ref{thm:USElp}) even though {\it SE} is
often the preferred definition of QIF in
practice~\cite{denning82,clarkjcs2007,malacaria:popl2007}.  This
suggests that approximation techniques may be necessary for checking
and inferring Shannon-entropy-based QIF.

\subsection{K-Observable Hyperproperties}

We define $k$-observable hyperproperty that refines the notion of
observable hyperproperties.  Informally, a $k$-observable
hyperproperty is a hyperproperty that can be verified by observing $k$
finite traces.
\begin{definition}[K-Observable Hyperproperties]
  We say that a hyperproperty $P$ is a $k$-observable hyperproperty if
  for any set of traces $S \in P$, there exists $T\in {\it Obs}$ such
  that $T\le S$, $|T|\le k$, and for any set of traces $S'\in{\it
    Prop}$, $T\le S'\Rightarrow S'\in P$.
\end{definition}
Clearly, any $k$-observable hyperproperty is an observable
hyperproperty.

We note that $k$-observable hyperproperties can be reduced to
$1$-observable hyperproperties by a simple program transformation
called {\em self composition}~\cite{barthe:csfw04,darvas:spc05}.
\begin{definition}[Parallel Self
  Composition~\cite{DBLP:journals/jcs/ClarksonS10}]
Parallel self composition of $S$ is defined as follows.
\[
S\times S = \aset{(s[0],s'[0]);(s[1],s'[1]);(s[2],s'[2]);\dots\mid s,s'\in S}
\]
where $s[i]$ denotes the $i$th element of $s$.
\end{definition}
Then, a $k$-product parallel self composition (simply self composition
henceforth) is defined as $S^k$.
\begin{theorem}
Every $k$-observable hyperproperty can be reduced to a $1$-observable
hyperproperty via a $k$-product self composition.
\end{theorem}
As an example, consider the following hyperproperty.  The
hyperproperty is the set of programs that return $1$ and $2$ for some
inputs.  Intuitively, the hyperproperty expresses two good things
happen (programs return $1$ and $2$) for programs.
\[
\aset{M\mid \exists h,h'. M(h)=1 \wedge M(h')=2}
\]
This is a $2$-observable hyperproperty as any program containing two
traces, one having $1$ as the output and the other having $2$ as the
output, satisfies it.

We can check the above property by self composition.  (Here, $||$
denotes a parallel composition.)
\[
\begin{array}{rcl}
  M'(H,H')&\;\equiv \;& O:=M(H)\; ||\; O':=M(H') \;||\;\sf{assert}(\neg(O=1\wedge O'=2))
\end{array}
\]
Clearly, $M$ satisfies the property iff the assertion failure is
reachable in the above program, that is, iff the predicate $O=1 \wedge
O'=2$ holds for some inputs $H, H'$.  (Note that, for convenience, we
take an assertion failure to be a ``good thing''.)

We show that neither the lower-bounding problem for min-entropy nor
the lower-bounding problem for guessing-entropy is a $k$-observable
hyperproperty for any $k$.
\begin{theorem}
\label{thm:LMELGEnkl}
  Neither $\mathcal{L}_{\it ME}$ nor $\mathcal{L}_{\it GE}$ is a
  $k$-observable property for any $k$.
\end{theorem}

However, if we let $q$ be a constant, then we obtain different
results.  First, we show that the lower-bounding problem for
min-entropy-based quantitative information flow under a constant bound
$q$, is a $\lfloor 2^q\rfloor +1$-observable hyperproperty.
\begin{theorem}
\label{thm:lmel}
Let $q$ be a constant.  Then, $\mathcal{L}_{\it ME}$ is a $\lfloor
2^q\rfloor +1$-observable hyperproperty.
\end{theorem}
The theorem follows from Lemma~\ref{lem:melog} below which states that
min-entropy based quantitative information flow under the uniform
distribution coincides with the logarithm of the number of output
traces.  That is, $(M,q)\in \mathcal{L}_{\it ME}$ iff there is an
evidence in $M$ containing $\lfloor 2^q\rfloor +1$ disjoint outputs.
\begin{lemma}[\cite{smith09}]
\label{lem:melog}
${\it ME}[U](M)=\log |\aset{o\mid \exists h. M(h)=o}|$
\end{lemma}

Next, we show that the lower-bounding problem for
guessing-entropy-based quantitative information flow under a constant
bound $q$ is a $\lfloor \frac{(\lfloor q\rfloor +1)^2}{\lfloor
  q\rfloor +1 -q}\rfloor +1$-observable hyperproperty.
\begin{theorem}
\label{thm:lgel}
  Let $q$ be a constant.  Then, $\mathcal{L}_{\it GE}$ is a $\lfloor
  \frac{(\lfloor q\rfloor +1)^2}{\lfloor q\rfloor +1 -q}\rfloor
  +1$-observable hyperproperty.
\end{theorem}
The proof of the theorem is similar to that of Theorem~\ref{thm:lmel},
in that the size of the evidence set can be computed from the bound
$q$.

\subsection{Computational Complexities}
\label{sec:complower}

We prove computational complexities of $\mathcal{L}_{\it ME}$ and
$\mathcal{L}_{\it GE}$ by utilizing their hyperproperty
classifications.  Following previous
work~\cite{DBLP:conf/csfw/YasuokaT10,DBLP:conf/esorics/YasuokaT10,DBLP:conf/csfw/CernyCH11},
we focus on boolean programs.

First, we introduce the syntax of boolean programs.  \addappendix{The
  semantics of boolean programs is standard and is deferred to
  Appendix (Figure~\ref{fig:semantics}).}{The semantics of boolean
  programs is standard.}  We call boolean programs without {\sf while}
statements {\em loop-free} boolean programs.
\begin{figure}[h]
\[
\begin{array}[t]{rcl}
  M & ::= &x:=\psi\mid M_0 ; M_1 
  \mid {\sf if}\; \psi\;{\sf then}\; M_0 \;{\sf else}\; M_1
  \mid {\sf while}\;\psi\;{\sf do}\;M\mid\ttskip \\
  \phi,\psi&::=&{\sf true}\mid x\mid \phi\wedge \psi\mid \neg \phi
\end{array}
\]
\caption{The syntax of boolean programs}
\label{fig:syntax}
\end{figure}

In this paper, we are interested in the computational complexity with
respect to the syntactic size of the input program (i.e., ``implicit
state complexity'', as opposed to \cite{DBLP:conf/csfw/CernyCH11}
which studies complexity over programs represented as explicit
states).

We show that the lower-bounding problems for min-entropy and
guessing-entropy are ${\it PP}$-hard.
\begin{theorem}
\label{thm:lmelgepp}
$\mathcal{L}_{\it ME}$ and $\mathcal{L}_{\it GE}$ for loop-free
boolean programs are {\it PP}-hard.
\end{theorem}
The theorem is proven by a reduction from ${\text{MAJSAT}}$, which is
a ${\it PP}$-hard problem.  ${\it PP}$ is the set of decision problems
solvable by a polynomial-time nondeterministic Turing machine which
accepts the input iff more than half of the computation paths accept.
{\text{MAJSAT}} is the problem of deciding, given a boolean formula
$\phi$ over variables $\overrightarrow x$, if there are more than
$2^{|\overrightarrow x|-1}$ satisfying assignments to $\phi$.

Next, we show that if $q$ be a constant, the upper-bounding problems
for min-entropy and guessing-entropy become {\it NP}-complete.
\begin{theorem}
\label{thm:lmelgenp}
  Let $q$ be a constant.  Then, $\mathcal{L}_{\it ME}$ and
  $\mathcal{L}_{\it GE}$ are {\it NP}-complete for loop-free boolean
  programs.
\end{theorem} 
{\it NP}-hardness is proven by a reduction from ${\text SAT}$, which
is a ${\it NP}$-complete problem.  The proof that $\mathcal{L}_{\it
  ME}$ and $\mathcal{L}_{\it GE}$ for a constant $q$ are in {\it NP}
follows from the fact that $\mathcal{L}_{\it ME}$ and
$\mathcal{L}_{\it GE}$ are $k$-observable hyperproperties for some
$k$.  We give the proof intuition for $\mathcal{L}_{\it ME}$.  Recall
that $k$-observable hyperproperties can be reduced to $1$-observable
hyperproperties via self composition.  Consequently, it is possible to
decide if the information flow of a given program $M$ is greater than
$q$ by checking if the predicate of the {\sf assert} statement is
violated for some inputs in the following program.
\[
\begin{array}{l}
  M'(H_1,H_2,\dots,H_n)\equiv\\
  \quad O_1:=M(H_1);O_2:=M(H_2);\dots;O_n:=M(H_n);\\
  \quad {\sf assert}(\bigvee_{i,j\in\aset{1,\dots,n}}(O_i=O_j\wedge i\not=j))
\end{array}
\]
where $n=\lfloor 2^q\rfloor +1$.  Let $\phi$ be the weakest
precondition of $O_1:=M(H_1);O_2:=M(H_2);\dots;O_n:=M(H_n)$ with
respect to the post condition
$\bigvee_{i,j\in\aset{1,\dots,n}}(O_i=O_j\wedge i\not=j)$.  Then,
${\it ME}[U](M)>q$ iff $\neg \phi$ is satisfiable.  Because a weakest
precondition of a loop-free boolean program is a polynomial size
boolean formula over the boolean variables representing the
inputs\footnote{For loop-free boolean programs, a weakest precondition
  can be constructed in polynomial
  time~\cite{DBLP:conf/popl/FlanaganS01,DBLP:journals/ipl/Leino05}.},
deciding ${\it ME}[U](M)>q$ is reducible to {\text SAT}.

For boolean programs (with loops), $\mathcal{L}_{\it ME}$ and
$\mathcal{L}_{\it GE}$ are {\it PSPACE}-complete, and
$\mathcal{L}_{\it SE}$ is {\it PSPACE}-hard (the tight upper-bound is
open for $\mathcal{L}_{\it SE}$).
\begin{theorem}
\label{thm:lmelgepspace}
  $\mathcal{L}_{\it ME}$ and $\mathcal{L}_{\it GE}$ are {\it
    PSPACE}-complete for boolean programs.
\end{theorem}

\begin{theorem}
\label{thm:lsepspace}
$\mathcal{L}_{\it SE}$ is {\it PSPACE}-hard for boolean programs.
\end{theorem}

\section{Safety Hyperproperties}

\label{sec:safety}

Clarkson and Schneider~\cite{DBLP:journals/jcs/ClarksonS10} have
proposed safety hyperproperties, a subclass of hyperproperties, as a
generalization of safety properties.  Intuitively, a safety
hyperproperty is a hyperproperty that can be refuted by observing a
finite set of finite traces.
\begin{definition}[Safety Hyperproperties~\cite{DBLP:journals/jcs/ClarksonS10}] We say that a
hyperproperty $P$ is a safety hyperproperty if for any set of traces
$S\not\in P$, there exists a set of traces $T\in {\it Obs}$ such that
$T\le S$, and $\forall S'\in{\it Prop}. T\le S'\Rightarrow S'\not\in
P$.
\end{definition}

We classify some upper-bounding problems as safety hyperproperties.
\begin{theorem}
\label{thm:umeugesp}
  $U_{\it ME}$ and $U_{\it GE}$ are safety hyperproperties.
\end{theorem}

Next, we review the definition of $k$-safety
hyperproperties~\cite{DBLP:journals/jcs/ClarksonS10}, which refines
the notion of safety hyperproperties.  Informally, a $k$-safety
hyperproperty is a hyperproperty which can be refuted by observing $k$
number of finite traces.
\begin{definition}[K-Safety
  Hyperproperties~\cite{DBLP:journals/jcs/ClarksonS10}]
  We say that a hyperproperty $P$ is a $k$-safety property if for any
  set of traces $S\not\in P$, there exists a set of traces $T\in {\it
    Obs}$ such that $T\le S$, $|T|\le k$, and $\forall S'\in{\it
    Prop}. T\le S'\Rightarrow S'\not\in P$.
\end{definition}
Note that $1$-safety hyperproperty is just the standard safety
property, that is, a property that can be refuted by observing a
finite execution trace.  The notion of $k$-safety hyperproperties
first came into limelight when it was noticed that non-interference
is a $2$-safety hyperproperty, but not a $1$-safety
hyperproperty~\cite{terauchi:sas05}.

A $k$-safety hyperproperty can be reduced to a $1$-safety
hyperproperty by self composition~\cite{barthe:csfw04,darvas:spc05}.
\begin{theorem}[\cite{DBLP:journals/jcs/ClarksonS10}]
  $k$-safety hyperproperty can be reduced to $1$-safety hyperproperty
  by self composition.
\end{theorem}

We have shown in our previous work that $\mathcal{U}_{\it ME}$ and
$\mathcal{U}_{\it GE}$ are $k$-safety hyperproperties when the bound
$q$ is fixed to a constant.
\begin{theorem}[\cite{DBLP:conf/esorics/YasuokaT10}]
\label{thm:umes}
Let $q$ be a constant.  ${\mathcal U}_{\it ME}$ is a $\lfloor
2^q\rfloor +1$-safety property.
\end{theorem}
\begin{theorem}[\cite{DBLP:conf/esorics/YasuokaT10}]
\label{thm:uges}
Let $q$ be a constant.  ${\mathcal U}_{\it GE}$ is a $\lfloor
\frac{(\lfloor q\rfloor +1)^2}{\lfloor q\rfloor +1 -q}\rfloor
+1$-safety property.
\end{theorem}

The only hyperproperty that is both a safety hyperproperty and a
liveness hyperproperty is $\mathcal{P}(\Psi_{\tt inf})$, that is, the
set of all traces~\cite{DBLP:journals/jcs/ClarksonS10}.  Consequently,
neither ${\mathcal U}_{\it ME}$ nor ${\mathcal U}_{\it GE}$ is a
liveness hyperproperty.

We have also shown in the previous work that the upper-bounding
problem for Shannon-entropy based quantitative information flow is not
a $k$-safety hyperproperty, even when $q$ is a constant.
\begin{theorem}[\cite{DBLP:conf/esorics/YasuokaT10}]
  Let $q$ be a constant.  ${\mathcal U}_{\it SE}$ is not a $k$-safety
  property for any $k>0$.
\end{theorem}

\subsection{Computational Complexities}
We prove computational complexities of upper-bounding problems by
utilizing their hyperproperty classifications.  As in
Section~\ref{sec:complower}, we focus on boolean programs.

First, we show that when $q$ is a constant, ${\it U}_{\it ME}$ and
${\it U}_{\it GE}$ are ${\it coNP}$-complete.
\begin{theorem}
\label{thm:umeugeconp}
  Let $q$ be a constant.  Then, $\mathcal{U}_{\it ME}$ and
  $\mathcal{U}_{\it GE}$ are {\it coNP}-complete for loop-free boolean
  programs.
\end{theorem}
{\it coNP}-hardness follows from the fact that non-interference is
{\it coNP}-hard~\cite{DBLP:conf/esorics/YasuokaT10}.  The {\it coNP}
part of the proof is similar to the {\it NP} part of
Theorem~\ref{thm:lmelgenp}, and uses the fact that $\mathcal{U}_{\it
  ME}$ is $k$-safety for a fixed $q$ and uses self composition.  By
self composition, the upper-bounding problem can be reduced to a
reachability problem (i.e., an assertion failure is unreachable for
any input).  To decide if ${\it ME}[U](M)\le q$, we construct the
following self-composed program $M'$ from the given program $M$.
\[
\begin{array}{l}
  M'(H_1,H_2,\dots,H_n)\equiv\\
  \quad O_1:=M(H_1);O_2:=M(H_2);\dots;O_n:=M(H_n);\\
  \quad {\sf assert}(\bigvee_{i,j\in\aset{1,\dots,n}}(O_i=O_j\wedge i\not=j))
\end{array}
\]
where $n=\lfloor 2^q\rfloor +1$.  Then, the weakest precondition of
$O_1:=M(H_1);O_2:=M(H_2);\dots;O_n:=M(H_n)$ with respect to the post
condition $\bigvee_{i,j\in\aset{1,\dots,n}}(O_i=O_j\wedge i\not=j)$ is
valid iff ${\it ME}[U](M)\le q$.  Because a weakest precondition of a
loop-free boolean program is a polynomial size boolean formula, and
the problem of deciding a given boolean formula is valid is a {\it
  coNP}-complete problem, $\mathcal{U}_{\it ME}$ is in {\it coNP}.

Like the lower-bounding problems $\mathcal{U}_{\it ME}$ and
$\mathcal{U}_{\it GE}$ for boolean programs (with loops) are {\it
  PSPACE}-complete, and $\mathcal{U}_{\it SE}$ is {\it PSPACE}-hard.
\begin{theorem}
\label{thm:umeugepspace}
  $\mathcal{U}_{\it ME}$ and $\mathcal{U}_{\it GE}$ are {\it
    PSPACE}-complete for boolean programs.
\end{theorem}
\begin{theorem}
\label{thm:usepspace}
$\mathcal{U}_{\it SE}$ is {\it PSPACE}-hard for boolean programs.
\end{theorem}

\section{Discussion}

\subsection{Bounding Domains}

\label{sec:bounddomain}
The notion of hyperproperty is defined over all programs regardless of
their size. (For example, non-interference is a $2$-safety property
for all programs and reachability is a safety property for all
programs.) But, it is easy to show that the lower bounding problems
would become ``$k$-observable'' hyperproperties if we constrained and
bounded the input domains because then the size of the semantics
(i.e., the number of traces) of such programs would be bounded by
$|\mathbb{H}|$ (and upper bounding problems would become
``$k$-safety''
hyperproperties~\cite{DBLP:conf/esorics/YasuokaT10}). In this case,
the problems are trivially $|\mathbb{H}|$-observable hyperproperties.
However, these bounds are high for all but very small domains, and are
unlikely to lead to a practical verification method.

\subsection{Observable Hyperproperties and Observable Properties}
As remarked in~\cite{DBLP:journals/jcs/ClarksonS10}, observable
hyperproperties generalize the notion of observable
properties~\cite{DBLP:journals/apal/Abramsky91}.  It can be shown that
there exists a non-empty observable property that is not a liveness
property (e.g., the set of all traces that starts with $\sigma$).  In
contrast, Theorem~\ref{thm:nonempty} states that every non-empty
observable hyperproperty is also a liveness hyperproperty.
Intuitively, this follows because the hyperproperty extension relation
$\le$ allows the right-hand side to contain traces that does not
appear in the left-hand side.  Therefore, for any $T\in {\it Obs}$,
there exists $T'\in{\it Prop}$ that contains $T$ and an evidence of
the observable hyperproperty.

\subsection{Maximum of QIF over Distribution}

\label{sec:maxqif}

Researchers have studied the maximum of QIF over the distribution.
For example, {\em channel
  capacity}~\cite{mccamant:pldi2008,malacaria08,NMS2009} is the
maximum of the Shannon-entropy based quantitative information flow
over the distribution (i.e., $\max_\mu{\it SE}[\mu]$).
Smith~\cite{smith09} showed that for any program without low-security
inputs, the channel capacity is equal to the min-entropy-based
quantitative information flow, that is, $\max_\mu{\it SE}[\mu] = {\it
  ME}[U]$.  Therefore, we obtain the same hyperproperty
classifications and complexity results for channel capacity as ${\it
  ME}[U]$.

{\it Min-entropy channel capacity} and {\it guessing-entropy channel
  capacity} are respectively the maximums of min-entropy based and
guessing-entropy based QIF over distributions (i.e., $\max_\mu{\it
  ME}[\mu]$ and $\max_\mu{\it GE}[\mu]$).  It has been shown that
$\max_\mu{\it ME}[\mu] = {\it
  ME}[U]$~\cite{DBLP:journals/entcs/BraunCP09,DBLP:conf/csfw/KopfS10}
and $\max_\mu{\it GE}[\mu] = {\it GE}[U]$~\cite{yasuoka:jocssubmit},
that is, they attain their maximums when the distributions are
uniform.  Therefore, they have the same hyperproperty classifications
and complexities as ${\it ME}[U]$ and ${\it GE}[U]$, which we have
already analyzed in this paper.

\section{Related Work}
{\v C}ern\'y et al.~\cite{DBLP:conf/csfw/CernyCH11} have investigated
the computational complexity of Shannon-entropy based QIF.  Formally,
they have defined a Shannon-entropy based QIF for interactive boolean
programs, and showed that the explicit-state computational complexity
of their lower-bounding problem is {\it PSPACE}-complete.  In
contrast, this paper's complexity results are ``implicit'' complexity
results of bounding problems of boolean programs (i.e., complexity
relative to the syntactic size of the input) some of which are
obtained by utilizing their hyperproperties classifications.

Clarkson and Schneider~\cite{DBLP:journals/jcs/ClarksonS10} have
classified quantitative information flow problems via hyperproperties.
Namely, they have shown that the problem of deciding if the channel
capacity of a given program is $q$, is a liveness hyperproperty.  And,
they have shown that an upper-bounding problem for the {\em
  belief}-based QIF~\cite{clarkson:csf2005} is a safety hyperproperty.
(It is possible to refine their result to show that their problem for
deterministic programs is actually equivalent to non-interference, and
therefore, is a $2$-safety hyperproperty~\cite{yasuoka:jocssubmit}.)

\section{Conclusion}

We have related the upper and lower bounding problems of quantitative
information flow, for various information theoretic definitions
proposed in literature, to Clarkson and Schneider's hyperproperties.
Hyperproperties generalize the classical trace properties, and are
thought to be more suitable for classifying information flow
properties as they are relations over sets of program traces.  Our
results confirm this by giving a fine-grained classification and
showing that it gives insights into the complexity of the QIF bounding
problems.  One of the contributions is a new class of hyperproperties:
{\em k-observable} hyperproperty.  We have shown that $k$-observable
hyperproperties are amenable to verification via self composition.

\bibliographystyle{eptcs}
\bibliography{klive}

\end{document}